\documentclass[11pt,reqno]{amsart}
\usepackage{mathrsfs}
\usepackage{amssymb}
\usepackage{amsmath}
\usepackage{amscd}
\begin{document}
\setlength{\oddsidemargin}{0cm} \setlength{\evensidemargin}{0cm}
\baselineskip=20pt

\theoremstyle{plain} \makeatletter
\newtheorem{theorem}{Theorem}[section]
\newtheorem{proposition}[theorem]{Proposition}
\newtheorem{lemma}[theorem]{Lemma}
\newtheorem{coro}[theorem]{Corollary}

\theoremstyle{definition}
\newtheorem{defi}[theorem]{Definition}
\newtheorem{notation}[theorem]{Notation}
\newtheorem{exam}[theorem]{Example}
\newtheorem{conj}[theorem]{Conjecture}
\newtheorem{prob}[theorem]{Problem}
\newtheorem{remark}[theorem]{Remark}
\newtheorem{claim}{Claim}

\numberwithin{equation}{section}

\title{Fermionic Novikov algebras admitting invariant non-degenerate symmetric bilinear forms are Novikov algebras}

\author{Zhiqi Chen}
\address{School of Mathematical Sciences and LPMC, Nankai University,
Tianjin 300071, P.R. China}\email{chenzhiqi@nankai.edu.cn}

\author{Ming Ding}
\address{Corresponding author. School of Mathematical Sciences and LPMC, Nankai University,
Tianjin 300071, P.R. China}\email{dingming@nankai.edu.cn}

\def\shorttitle{Fermionic Novikov algebras}

\subjclass[2010]{17B60, 17A30, 17D25}


\begin{abstract}
This paper is to prove that a fermionic Novikov algebra equipped with an invariant non-degenerate symmetric bilinear form is a Novikov algebra.
\end{abstract}
\maketitle


\renewcommand{\theequation}{\arabic{section}.\arabic{equation}}
\section{Introduction}
Gel$^\prime$fand and Dikii gave a bosonic formal variational
calculus in \cite{GD1,GD2} and Xu  gave a fermionic formal
variational calculus in \cite{Xu1}. Combining the
bosonic theory of Gel$^\prime$fand-Dikii and the fermionic theory, Xu gave in \cite{Xu2} a formal variational calculus of
super-variables. Fermionic Novikov algebras are related to the Hamiltonian
super-operator in terms of this theory. A fermionic Novikov algebra is a finite-dimensional vector space $A$ over a field $\mathbb F$ with a bilinear
product $(x,y)\mapsto xy$ satisfying
\begin{equation}\label{1}
(xy)z-x(yz)=(yx)z-y(xz),
\end{equation}
\begin{equation}\label{rsy}
(xy)z=-(xz)y
\end{equation}
for any $x,y,z\in A$. It corresponds to the following Hamiltonian
operator $H$ of type 0 (\cite{Xu2}):
\begin{equation}
H^0_{\alpha,\beta}=\sum\limits_{\gamma\in
I}(a^\gamma_{\alpha,\beta}\Phi_\gamma(2)+b^\gamma_{\alpha,\beta}\Phi_\gamma
D),\quad
a^\gamma_{\alpha,\beta},b^\gamma_{\alpha,\beta}\in\mathbb{R}.
\end{equation}

Fermionic Novikov algebras are a class of left-symmetric
algebras which are defined by the identity~(\ref{1}). Left-symmetric algebras are a class of non-associative algebras arising from the
study of affine manifolds, affine structures and convex homogeneous
cones (\cite{Bu1,Vi1}). Novikov algebras are another class of left-symmetric algebras $A$ satisfying
\begin{equation}\label{Novi}
  (xy)z=(xz)y, \quad\forall x,y,z\in A
\end{equation}
Novikov algebras were introduced in connection with the Poisson
brackets of hydrodynamic type (\cite{B-N1,DN1,DN2}) and Hamiltonian
operators in the formal variational calculus
(\cite{GD1,GD2,GD3,Xu1,Xu3}).

The commutator of a left-symmetric algebra $A$
\begin{equation}\label{com}
[x,y]=xy-yx
\end{equation}
defines a Lie algebra, which is called the underlying Lie algebra of $A$. A bilinear form $\langle\cdot,\cdot\rangle$ on a left-symmetric algebra $A$ is invariant if \begin{equation}
\langle R_{x}y, z\rangle=\langle
y, R_{x}z\rangle \end{equation}
for any $x, y, z\in A$.

Zelmanov (\cite{Ze1}) classifies real Novikov algebras with invariant positive definite symmetric bilinear forms. In \cite{MG}, Guediri gives the classification for the Lorentzian case. This paper is to study real fermionic Novikov algebras admitting invariant non-degenerate symmetric bilinear forms. The main result is the following theorem.
\begin{theorem}\label{maintheorem}
Any finite dimensional real fermionic Novikov algebra admitting an invariant non-degenerate symmetric bilinear form is a Novikov algebra.
\end{theorem}
In order to prove Theorem~\ref{maintheorem}, we describe the structure of these fermionic Novikov algebras. But we only give part of the classification since the complete classification is very complicated.

\section{The proof of Theorem~\ref{maintheorem}}
Let $A$ be a fermionic Novikov algebra and let $L_x$ and $R_{x}$ denote the left and right multiplication operator by the element $x\in A$ respectively, i.e., $$L_{x}(y)=xy,\quad R_{x}(y)=yx$$ for any $y\in A$. By the equation~(\ref{rsy}), we have $$R_{x}R_{y}=-R_{y}R_{x},\quad \forall x,y\in A.$$ In particular, $R_{x}^2=0$ for any $x\in A$.

\begin{defi}\label{defi}
A non-degenerate bilinear form $\langle\cdot,\cdot\rangle$ on $V$ is of type $(n-p,p)$ if there is a basis $\{e_1,\ldots,e_n\}$ of $V$ such that
$\langle e_i,e_i\rangle=-1$ for $1\leq i\leq p$, $\langle e_i,e_i\rangle=1$ for $p+1\leq i\leq n$, and $\langle e_i,e_j\rangle=0$ for otherwise. The bilinear form is positive definite if $p=0$; Lorentzian if $p=1$.
\end{defi}

A linear operator $\sigma$ of $(V,\langle\cdot,\cdot\rangle)$ is self-adjoint if $$\langle \sigma(x), y\rangle=\langle
x, \sigma(y)\rangle,\quad \forall x,y\in V.$$

\begin{lemma}[\cite{O}, pp. 260-261]\label{mainlemma}
A linear operator $\sigma$ on $V=\mathbb{R}^{n}$ is self-adjiont if and only if $V$ can be expressed as a direct sum of $V_{k}$ that are mutually orthogonal
(hence non-degenerate), $\sigma-$invariant, and each $\sigma\mid_{V_k}$ has a $r\times r$ matrix form either $$\left(  \begin{array}{cccc} \lambda & 0 & \cdots & 0 \\
 1 & \lambda & \cdots & \vdots  \\
  \vdots & \ddots &  \lambda & 0  \\
  0 & \cdots & 1 & \lambda
   \end{array} \right)$$
relative to a basis $\alpha_1,\ldots,\alpha_r(r\geq 1)$   with all scalar products  zero except $\langle \alpha_i,\alpha_j\rangle=\pm 1$ if $i+j=r+1$,
or
   $$\left(  \begin{array}{cccccc} \left(
                                            \begin{array}{cc}
                                              a & b \\
                                              -b & a \\
                                            \end{array}
                                          \right)
    &  &  &  &  &  \\
  \left(
                                            \begin{array}{cc}
                                              1 & 0 \\
                                              0 & 1 \\
                                            \end{array}
                                          \right)
    & \left(
                                            \begin{array}{cc}
                                              a & b \\
                                              -b & a \\
                                            \end{array}
                                          \right) &  &  & 0 &   \\

    & \left(
                                            \begin{array}{cc}
                                              1 & 0 \\
                                              0 & 1 \\
                                            \end{array}
                                          \right) & \left(
                                            \begin{array}{cc}
                                              a & b \\
                                              -b & a \\
                                            \end{array}
                                          \right) &  &  &  \\

    0 &  &\ddots  &  & \ddots &  \\

     &  &  & & \left(
                 \begin{array}{cc}
                   1 & 0 \\
                   0 & 1 \\
                 \end{array}
               \right)
      & \left(
                 \begin{array}{cc}
                   a & b \\
                   -b & a \\
                 \end{array}
               \right)
   \end{array} \right)$$
where $b\neq 0$ relative to a basis $\beta_1,\alpha_1,\ldots,\beta_m,\alpha_m$   with all scalar products  zero
except $\langle \beta_i,\beta_j\rangle=1=-\langle \alpha_i,\alpha_j\rangle$ if $i+j=m+1$.
\end{lemma}

If $A$ admits an invariant non-degenerate symmetric bilinear form $\langle\cdot,\cdot\rangle$ of type $(n-p,p)$, then $-\langle\cdot,\cdot\rangle$ is an invariant non-degenerate symmetric bilinear form on $A$ of type $(p,n-p)$. So we can assume $p\leq n-p$.

\begin{lemma}\label{lemma-iso}
Let $A$ be a fermionic Novikov algebra admitting an invariant non-degenerate symmetric bilinear form $\langle\cdot,\cdot\rangle$ of type $(n-p,p)$, then $\dim ImR_{x}\leq p$ for any $x\in A$.
\end{lemma}
\begin{proof}
Recall that $R_{x}^{2}=0$, it follows that $ImR_{x}\subseteq KerR_{x}$.
By the invariance of $\langle\cdot,\cdot\rangle$, we have $\langle R_{x}y, R_{x}z\rangle=\langle y, R_{x}^{2}z\rangle=0$ which yields $\langle ImR_{x},ImR_{x}\rangle=0$. Hence $\dim ImR_{x}\leq p$.
\end{proof}

Let $x_{0}\in A$ satisfy $\dim ImR_{x}\leq \dim ImR_{x_0}$ for any $x\in A$. By Lemma $\ref{lemma-iso}$, $\dim ImR_{x_0}\leq p$. For convenience, let $\dim ImR_{x_0}=k$.
By Lemma $\ref{mainlemma}$ and $R_{x_0}^{2}=0$, there exists a basis $\{e_1,\ldots,e_n\}$ of $A$ such that the operator $R_{x_0}$ relative to the basis has the matrix of form
$$\left(
    \begin{array}{ccc}
      \left(
        \begin{array}{ccc}
          \left(
            \begin{array}{cc}
              0 & 0 \\
              1 & 0 \\
            \end{array}
          \right)
           & 0 &  \\
           & \ddots &  \\
           & 0 & \left(
            \begin{array}{cc}
              0 & 0 \\
              1 & 0 \\
            \end{array}
          \right) \\
        \end{array}
      \right)_{2k\times 2k}
       & 0_{2k\times (n-2k)} \\
      0_{(n-2k)\times 2k} & 0_{(n-2k)\times (n-2k)} \\
    \end{array}
  \right),
$$
where the matrix of the metric $\langle\cdot,\cdot\rangle$ with respect to $\{e_1,\ldots,e_n\}$ is $$\left(
    \begin{array}{ccc}
      C_{2k} & 0 & 0 \\
      0 & -I_{p-k} & 0 \\
      0 & 0 & I_{n-p-k}
    \end{array}
  \right).$$
Here $C_{2k}={\rm diag}\left(\left(
           \begin{array}{cc}
              0 & 1 \\
              1 & 0
            \end{array}
          \right), \cdots,
          \left(
            \begin{array}{cc}
              0 & 1 \\
              1 & 0
            \end{array}
          \right)
              \right)$ and $I_{s}$ denotes the $s\times s$ identity matrix.
For any $x\in A$, the matrix of the operator $R_{x}$ relative to the basis is
$$\left(
    \begin{array}{ccc}
      A_1 & A_2 & A_3 \\
      A_4 & A_5 & A_6 \\
      A_7 & A_8 & A_9 \\
    \end{array}
  \right),
$$
whose blocks are the same as those of the metric matrix under the basis $\{e_1,\ldots,e_n\}$.

Firstly we can prove that $$\left(
              \begin{array}{cc}
                A_5 & A_6 \\
                A_8 & A_9 \\
              \end{array}
            \right)=0_{(n-2k)\times (n-2k)}.$$
In fact, assume that there exists some nonzero entry of $\left(
              \begin{array}{cc}
                A_5 & A_6 \\
                A_8 & A_9 \\
              \end{array}
            \right)$
which denoted by $d$. Consider the matrix form of the operator $R_{x}+lR_{x_0}$. With no confusions, we do not distinguish between the operator $R_x$ and  its matrix form in the following. For any $l\in \mathbb{R}$, by the choice of $x_0$, we know that $r(R_{x}+lR_{x_0})=r(R_{x+lx_0})\leq k$. Taking 2nd, $\cdots$, $2k$-th rows, 1st, $\cdots$, (2k-1)-th columns, and the row and column containing the element $d$ in the matrix of $R_{x}+lR_{x_0}$, we have the $(k+1)\times (k+1)$ matrix $\left(
          \begin{array}{cc}
            B+lI_{k} & \alpha \\
            \beta & d \\
          \end{array}
        \right)$. Note that the determinant of
$\left(
          \begin{array}{cc}
           B+lI_{k} & \alpha \\
            \beta & d \\
          \end{array}
        \right)$, i.e., $$\left|
          \begin{array}{cc}
           B+lI_{k} & \alpha \\
            \beta & d \\
          \end{array}
        \right|,$$
is a polynomial of degree $k$ in a single indeterminate $l$. So we can choose some $l'\in \mathbb{R}$ such that the determinant is nonzero. It follows that $$r(R_{x}+l'R_{x_0})=r(R_{x+l'x_0})\geq k+1,$$ which is a contradiction.

Secondly, by  $R_{x}R_{x_0}+R_{x_0}R_{x}=0$, we have that $A_1=(M_{ij})_{k\times k}$ where $M_{ij}=\left(
                                                                                                       \begin{array}{cc}
                                                                                                         b_{ij} & 0 \\
                                                                                                         d_{ij} & -b_{ij} \\
                                                                                                       \end{array}
                                                                                                     \right)$,
$$A_2=\left(
             \begin{array}{cccc}
               0 & \cdots & \cdots & 0 \\
               a_{2,1} & \cdots & \cdots & a_{2,p-k} \\
               \vdots & \vdots & \vdots & \vdots \\
               0 & \cdots & \cdots & 0  \\
               a_{2k,1} & \cdots & \cdots & a_{2k,p-k} \\
             \end{array}
           \right),$$ $$A_3=\left(
             \begin{array}{cccc}
               0 & \cdots & \cdots & 0 \\
               c_{2,1} & \cdots & \cdots & c_{2,n-p-k} \\
               \vdots & \vdots & \vdots & \vdots \\
               0 & \cdots & \cdots & 0  \\
               c_{2k,1} & \cdots & \cdots & c_{2k,n-p-k} \\
             \end{array}
           \right)
.$$
Furthermore, since $\langle R_{x}y, z\rangle=\langle y, R_{x}z\rangle$, we obtain that $$M_{ij}=\left(
                                                                                                       \begin{array}{cc}
                                                                                                         b_{ij} & 0 \\
                                                                                                         d_{ij} & -b_{ij} \\
                                                                                                       \end{array}
                                                                                                     \right), M_{ji}=\left(
                                                                                                       \begin{array}{cc}
                                                                                                         -b_{ij} & 0 \\
                                                                                                         d_{ij} & b_{ij} \\
                                                                                                       \end{array} \right),$$
where $b_{ii}=0$ for any $1\leq i\leq k$, and
$$A_4=-\left(
                                                            \begin{array}{ccccc}
                                                              a_{2,1} & 0 & \cdots & a_{2k,1} & 0 \\
                                                              \vdots & \vdots & \vdots & \vdots & \vdots \\
                                                              \vdots & \vdots & \vdots & \vdots & \vdots \\
                                                              a_{2,p-k} & 0 & \cdots & a_{2k,p-k} & 0 \\
                                                            \end{array}
                                                          \right), $$$$
A_7=\left(
                                                            \begin{array}{ccccc}
                                                              c_{2,1} & 0 & \cdots & c_{2k,1} & 0 \\
                                                              \vdots & \vdots & \vdots & \vdots & \vdots \\
                                                              \vdots & \vdots & \vdots & \vdots & \vdots \\
                                                              c_{2,n-p-k} & 0 & \cdots & c_{2k,n-p-k} & 0 \\
                                                            \end{array}
                                                           \right). $$
 Since $R_{x}^{2}=0$, we have that $A_{1}^{2}+A_{2}A_{4}+A_{3}A_{7}=0_{2k\times 2k}.$
Note that $$0=(A_{1}^{2}+A_{2}A_{4}+A_{3}A_{7})_{i,i}=(A_{1}^{2})_{i,i}.$$ It follows that $b_{ij}=0$ for any $i,j$. Then $$M_{ij}=M_{ji}=\left(
                                                                                                       \begin{array}{cc}
                                                                                                         0 & 0 \\
                                                                                                         d_{ij} & 0 \\
                                                                                                       \end{array}
                                                                                                     \right).$$

Finally, we claim that  $A_2,A_3,A_4$ and $A_7$ are zero matrices. Here we only prove $A_2=0_{2k\times (p-k)}$, similar for others. Assume that there exists some nonzero entry of $A_2$ which denoted by $d$. Consider the matrix of the operator $R_{x}+lR_{x_0}$. Similar to the proof of $$\left(
              \begin{array}{cc}
                A_5 & A_6 \\
                A_8 & A_9 \\
              \end{array}
            \right)=0_{(n-2k)\times (n-2k)},$$ we consider the matrix
 $$ \left(
          \begin{array}{cc}
            A_{1}^{'}+lI_{k} & \alpha^{T} \\
            -\alpha & 0 \\
          \end{array}
        \right),
 $$ where $d$ is an entry in the vector $\alpha$ and $A_{1}^{'}=(d_{ij})_{k\times k}$ is a symmetric matrix. Thus there exists an orthogonal matrix $P$ such that
  $P^{T}A_{1}^{'}P=\left(
                     \begin{array}{ccc}
                       \lambda_1 &  & 0 \\
                        & \ddots &  \\
                       0 &  & \lambda_k \\
                     \end{array}
                   \right).
  $ Choose some $l> \{|\lambda_1|,\cdots,|\lambda_k|\}$. Then the matrix $A_{1}^{'}+lI_{k}$ is invertible. We have
  $$\left|
          \begin{array}{cc}
            A_{1}^{'}+lI_{k} & \alpha^{T} \\
            -\alpha & 0 \\
          \end{array}
        \right|=\left|\left(
          \begin{array}{cc}
            P^{T} & 0 \\
            0 & 1 \\
          \end{array}
        \right)\left(
          \begin{array}{cc}
            A_{1}^{'}+lI_{k} & \alpha^{T} \\
            -\alpha & 0 \\
          \end{array}
        \right)\left(
          \begin{array}{cc}
            P & 0 \\
            0 & 1 \\
          \end{array}
        \right)\right|$$$$=\left|
                   \begin{array}{cc}
                     \left(
                     \begin{array}{ccc}
                       \lambda_1+l &  & 0 \\
                        & \ddots &  \\
                       0 &  & \lambda_k+l \\
                     \end{array}
                   \right) & \beta^{T} \\
                     -\beta & 0 \\
                   \end{array}
                 \right|=(\Pi_{i=1}^{k}(\lambda_i+l))\Sigma_{i=1}^{k}\frac{1}{\lambda_i+l}b_{i}^{2}\neq 0, $$
 where $\beta=\alpha P=(b_{1},\cdots,b_{k})$ is a nonzero vector. It follows that $$r(R_{x}+lR_{x_0})
=r(R_{x+lx_0})\geq k+1,$$ which is a contradiction. That is, $A_2=0_{2k\times (p-k)}$.

Up to now, we know that the matrix of $R_x$ is
$$\left(
    \begin{array}{cc}
      A_1 & 0_{2k\times (n-2k)} \\
      0_{(n-2k)\times 2k} & 0_{(n-2k)\times (n-2k)} \\
    \end{array}
  \right),
$$
where $A_1=(M_{ij})_{k\times k}$, here $M_{ij}=M_{ji}=\left(
\begin{array}{cc}
 0 & 0 \\
  d_{ij}(x) & 0 \\
   \end{array}
    \right).$ Hence $R_{x}R_{y}=0$ for any $x,y\in A$, which implies Theorem~\ref{maintheorem}.

\section{The structure of such fermionic Novikov algebras}
Let $A$ be an $n$-dimensional fermionic Novikov algebra with an invariant non-degenerate symmetric bilinear form $\langle\cdot,\cdot\rangle$ of type $(n-p,p)$.
By the above section, if $x_{0}\in A$ satisfies $$\dim ImR_{x}\leq \dim ImR_{x_0}=k\leq p$$ for any $x\in A$, then there exists a basis $\{e_1,\cdots,e_n\}$ such that the matrix of $R_x$ is
$$\left(
    \begin{array}{cc}
      A_1 & 0_{2k\times (n-2k)} \\
      0_{(n-2k)\times 2k} & 0_{(n-2k)\times (n-2k)} \\
    \end{array}
  \right),
$$
where $A_1=(M_{ij})_{k\times k}$, here $M_{ij}=M_{ji}=\left(
\begin{array}{cc}
 0 & 0 \\
  d_{ij}(x) & 0 \\
   \end{array}
    \right)$. In particular, $d_{ii}(x_0)=1$ for $i=1,\cdots,k$ and others zero. Clearly
\begin{proposition}
$\dim AA=\dim ImR_{x_0}=k.$
\end{proposition}

If $k=0$, then $xy=0$ for any $x,y\in A$.

If $k=1$, then there exists a basis $\{e_1,\cdots,e_n\}$ such that the matrix of $R_x$ is
$$\left(
    \begin{array}{cc}
      M & 0_{2\times (n-2)} \\
      0_{(n-2)\times 2} & 0_{(n-2)\times (n-2)} \\
    \end{array}
  \right),
$$
where $M=\left(
\begin{array}{cc}
 0 & 0 \\
  d(x) & 0 \\
   \end{array}
    \right)$.
Clearly the matrices of $L_{e_i}$ are zero matrices if $i\not=1$. Thus $$L_xL_y=L_yL_x,\quad \forall x,y\in A.$$
Together with $R_{x}R_{y}=0$ for any $x,y\in A$, the matrices of $R_{e_i}$ for $1\leq i\leq n$ determine a fermionic Novikov algebra. Furthermore
$A$ is one of the following cases:
\begin{enumerate}
   \item $k=1$, and there exists a basis $\{e_1,\cdots,e_n\}$ such that $e_1e_1=e_2$ and others zero.
   \item $k=1$, and there exists a basis $\{e_1,\cdots,e_n\}$ such that $e_1e_2=e_2$ and others zero.
   \item $k=1$, and there exists a basis $\{e_1,\cdots,e_n\}$ such that $e_1e_3=e_2$ and others zero.
\end{enumerate}
In particular, the above discussion gives the classification of fermionic Novikov algebras admitting invariant Lorentzian symmetric bilinear forms which is obtained in \cite{MG}.

If $k=2$, then there exists a basis $\{e_1,\cdots,e_n\}$ such that nonzero products are given by
   $$e_1e_i=\lambda_ie_2+\mu_ie_4,\quad e_3e_i=\mu_ie_2+\gamma_ie_4.$$
For this case, $A$ is a fermionic Novikov algebra if and only if $L_{e_1}L_{e_3}=L_{e_3}L_{e_1}$. But the complete classification is very complicated. It is similar for $k\geq 3$.

\section{Acknowledgements}
This work was supported by NSF of China (No. 11301282) and Specialized Research Fund for the
Doctoral Program of Higher Education (No. 20130031120004).

\end{document}